\documentclass{amsart}
\usepackage{amsmath,amsfonts,amssymb}
\usepackage{graphicx,amsmath, amscd}
\usepackage{hyperref}
\usepackage{color,comment}
\usepackage[T1]{fontenc}
\usepackage[latin1]{inputenc}

\usepackage[lmargin=2.2cm,rmargin=2.2cm,tmargin=2.3cm,bmargin=2.3cm]{geometry}
\usepackage[toc,page]{appendix}

\newtheorem{theorem}{Theorem}[section]
\newtheorem{lemma}[theorem]{Lemma}
\newtheorem{corollary}[theorem]{Corollary}

\newtheorem{definition}[theorem]{Definition}

\newtheorem{defi/prop}[theorem]{Definition/Proposition}

\newcommand{\N}{\mathbb{N}}

\newcommand{\C}{\mathbb{C}}
\renewcommand{\P}{\mathbf{P}}
\newcommand{\E}{\mathbf{E}}
\newcommand{\Cov}{{\mathbf{Cov}}}

\newcommand{\overbar}[1]{\mkern 1.5mu\overline{\mkern-1.5mu#1\mkern-1.5mu}\mkern 1.5mu}

\renewcommand{\leq}{\leqslant}
\renewcommand{\geq}{\geqslant}

\DeclareMathOperator{\tr}{Tr}

\newcommand{\ketbra}[2]{| #1 \rangle\!\langle #2 |}
\newcommand{\bra}[1]{\langle #1 |}
\newcommand{\ket}[1]{| #1 \rangle}

\title{Optimal quantum (tensor product) expanders from unitary designs}

\date{November 4 2025}

\author{C\'{e}cilia Lancien}

\address{\textbf{C\'{e}cilia Lancien:} CNRS \& Institut Fourier, Universit\'e Grenoble Alpes, 100 rue des maths, 38610 Gi\`eres, France.} 
\email{cecilia.lancien@univ-grenoble-alpes.fr}

\begin{document}
	
	\begin{abstract} In this work we investigate how quantum expanders (i.e.~quantum channels with few Kraus operators but a large spectral gap) can be constructed from unitary designs. Concretely, we prove that a random quantum channel whose Kraus operators are independent unitaries sampled from a $2$-design measure is with high probability an optimal expander (in the sense that its spectral gap is as large as possible). More generally, we show that, if these Kraus operators are independent unitaries of the form $U^{\otimes k}$, with $U$ sampled from a $2k$-design measure, then the corresponding random quantum channel is typically an optimal $k$-copy tensor product expander, a concept introduced by Harrow and Hastings (Quant.~Inf.~Comput.~2009).
	\end{abstract}

\maketitle	
	
\section{Preliminary facts and tools}

\subsection{Brief reminder on quantum channels and their spectral properties} \hfill\vspace{0.1cm}

A quantum channel on $\mathcal M_N(\C)$ is described by a linear map $\Phi:\mathcal M_N(\C)\to\mathcal M_N(\C)$ that is completely positive (CP) and trace-preserving (TP). We recall that the action of a CP map $\Phi$ on $\mathcal M_N(\mathbb C)$ can always be written in the following (non-unique) way, called a Kraus representation (see e.g.~\cite[Section 2.3.2]{Aubrun17}):
\begin{equation} \label{eq:Kraus}
	\Phi:X\in\mathcal M_N(\mathbb C) \mapsto \sum_{s=1}^d K_sXK_s^*\in\mathcal M_N(\mathbb C), 
\end{equation}
for some $d\in\N$ and some $K_1,\ldots,K_d\in\mathcal M_N(\mathbb C)$, called Kraus operators of $\Phi$. The fact that $\Phi$ is TP is equivalent to the following constraint on the $K_s$'s:
\[ \sum_{s=1}^d K_s^*K_s = I. \]
The smallest $d$ such that an expression of the form of equation \eqref{eq:Kraus} for $\Phi$ exists is called the Kraus rank of $\Phi$. The Kraus rank of a CP map on $\mathcal M_N(\mathbb C)$ is always at most $N^2$.

Given a linear map $\Phi:\mathcal M_N(\C)\to\mathcal M_N(\C)$, we denote by $\lambda_1(\Phi),\ldots,\lambda_{N^2}(\Phi)$, resp.~$s_1(\Phi),\ldots,s_{N^2}(\Phi)$, its eigenvalues, resp.~singular values, ordered so that $|\lambda_1(\Phi)|\geq\cdots\geq|\lambda_{N^2}(\Phi)|$, resp.~$s_1(\Phi)\geq\cdots\geq s_{N^2}(\Phi)$. A quantum channel $\Phi$ on $\mathcal M_N(\mathbb C)$ always has a largest (in modulus) eigenvalue that is equal to $1$ (see e.g.~\cite[Proposition 6.1]{Wolf12}), which we will denote by convention $\lambda_1(\Phi)$, and consequently a largest singular value $s_1(\Phi)$ that is at least $1$ . In addition, the eigenvalue $1$ has an associated eigenvector which is a positive semidefinite matrix (see e.g.~\cite[Theorem 6.5]{Wolf12}). Hence $\Phi$ always has a fixed state, i.e.~a quantum state $\hat\rho$ such that $\Phi(\hat\rho)=\hat\rho$. $\Phi$ is said to be unital if one of its fixed states is the so-called maximally mixed state $I/N$. It turns out that, if $\Phi$ is unital, then $s_1(\Phi)$ is also equal to $1$ (see e.g.~\cite[Theorem 4.27]{Watrous18}).

The constraint of being unital is dual to that of being TP. What we mean is that a CP map $\Phi$ is unital if and only if its dual (or adjoint) CP map $\Phi^*$, for the Hilbert-Schmidt inner product, is TP. It thus reads at the level of Kraus operators as 
\[ \sum_{s=1}^d K_sK_s^* = I. \]

Note that, identifying $\mathcal M_N(\mathbb C)$ with $\mathbb C^N\otimes\mathbb C^N$, a linear map $\Phi:\mathcal M_N(\mathbb C)\to\mathcal M_N(\mathbb C)$ can equivalently be seen as a linear map $M_\Phi:\mathbb C^N\otimes\mathbb C^N\to\mathbb C^N\otimes\mathbb C^N$, i.e.~an element of $\mathcal M_{N^2}(\mathbb C)$. Concretely, a CP linear map
\[ \Phi:X\in\mathcal M_N(\mathbb C) \mapsto \sum_{s=1}^d K_sXK_s^*\in\mathcal M_N(\mathbb C) \]
can be identified with
\[ M_\Phi = \sum_{s=1}^d K_s\otimes\overbar K_s \in \mathcal M_{N^2}(\mathbb C), \]
where $\overbar K_s$ denotes the entry-wise conjugate of $K_s$. This identification preserves the eigenvalues and singular values, i.e.~$\lambda_i(\Phi)=\lambda_i(M_\Phi)$ and $s_i(\Phi)=s_i(M_\Phi)$ for each $1\leq i\leq N^2$, and hence the Schatten norms as well.

Given $\Phi$ a unital quantum channel on $\mathcal M_N(\mathbb C)$, let us introduce a few extra notations related to its spectral properties. We first define
\begin{equation} \label{eq:def-lambda} r(\Phi) := \max\left\{ 1\leq i\leq N^2 : |\lambda_i(\Phi)|=1 \right\} \quad \text{and} \quad \lambda(\Phi):=|\lambda_{r(\Phi)+1}(\Phi)|, \end{equation}
so that $\lambda_1(\Phi),\ldots,\lambda_{r(\Phi)}(\Phi)$ are the so-called peripheral eigenvalues of $\Phi$ (i.e.~those with modulus $1$) and $\lambda(\Phi)<1$ is the largest modulus of a non-peripheral eigenvalue of $\Phi$. It is known that peripheral eigenvalues have trivial Jordan blocks (see e.g.~\cite[Proposition 6.2]{Wolf12}). This means that $\Phi$ can always be written as 
\[ \Phi = \sum_{i=1}^{r(\Phi)} \lambda_i(\Phi)\Pi^{(i)}_\Phi + \Delta_\Phi, \]
where the $\Pi^{(i)}_\Phi$'s are (not necessarily orthogonal) rank $1$ projectors satisfying $\Pi^{(i)}_\Phi\Pi^{(j)}_\Phi=\delta_{i,j}\Pi^{(i)}_\Phi$ and $\Delta_\Phi\Pi^{(i)}_\Phi=\Pi^{(i)}_\Phi\Delta_\Phi=0$. We can similarly define
\begin{equation} \label{eq:def-s} r'(\Phi) := \max\left\{ 1\leq i\leq N^2 : s_i(\Phi)=1 \right\} \quad \text{and} \quad s(\Phi):=s_{r'(\Phi)+1}(\Phi), \end{equation}
so that $s(\Phi)<1$ is the largest non-$1$ singular value of $\Phi$.
It is easy to see that $r'(\Phi)\geq r(\Phi)$, and consequently that $s(\Phi)\leq s_{r(\Phi)+1}(\Phi)$. Indeed, by Weyl's majorant theorem (see e.g.~\cite[Theorem~II.3.6]{Bhatia97}), we have $|\lambda_1(\Phi)|+\cdots+|\lambda_{r(\Phi)}(\Phi)|\leq s_1(\Phi)+\cdots+s_{r(\Phi)}(\Phi)$, and by definition the left-hand side is equal to $r(\Phi)$ while the right-hand side is at most $r(\Phi)$, with equality if and only if $s_{r(\Phi)}(\Phi)=1$ i.e.~if and only if $r'(\Phi)\geq r(\Phi)$.

Setting
\begin{equation} \label{eq:def-Pi} \Pi_\Phi := \sum_{i=1}^{r(\Phi)} \lambda_i(\Phi)\Pi^{(i)}_\Phi, \end{equation}
it can be shown that $\Pi_\Phi$ is actually a unital quantum channel on $\mathcal M_N(\C)$ (see e.g.~\cite[Proposition 6.3]{Wolf12}). We then have by definition $\lambda(\Phi)=|\lambda_{r(\Phi)+1}(\Phi)|=|\lambda_1(\Phi-\Pi_\Phi)|$ and by the min-max principle for singular values (see e.g.~\cite[Corollary~III.1.2 and Problem~III.6.1]{Bhatia97}) $s(\Phi)\leq s_{r(\Phi)+1}(\Phi)\leq s_1(\Phi-\Pi_\Phi)$.

With these observations at hand, we are now ready to give a first general definition of what it means to be an expander for a unital quantum channel. We will later give an apparently slightly different one, suited to the more specific case that we will actually focus on (cf.~Definition \ref{def:expander}). But as we will see, both actually coincide.

\begin{definition}[Expander] \label{def:expander-general}
	Fix $0<\lambda<1$. A unital quantum channel $\Phi$ on $\mathcal M_N(\mathbb C)$ with Kraus rank $d$ is called a $(d,\lambda)$ expander if it satifies one of the following properties: (1) $\lambda(\Phi) \leq \lambda$ or (2) $s(\Phi) \leq \lambda$, where the parameters $\lambda(\Phi)$ and $s(\Phi)$ are as defined in equations \eqref{eq:def-lambda} and \eqref{eq:def-s} respectively.
\end{definition}

Let us explain the meaning of both conditions. Given a unital quantum channel $\Phi$, the parameters $\lambda(\Phi)$ and $s(\Phi)$ can be seen as quantifying how far $\Phi$ is from the quantum channel $\Pi_\Phi$, defined in equation \eqref{eq:def-Pi}, that acts as $\Phi$ on the peripheral eigenspace of $\Phi$ and cancels its other eigenspaces. Indeed, we have already explained that $\lambda(\Phi)=|\lambda_1(\Phi-\Pi_\Phi)|$ and $s(\Phi)\leq s_1(\Phi-\Pi_\Phi)$, with the latter inequality in fact being an equality for the unital quantum channels that we will look at later on. If this is so, we see that, the smaller $\lambda(\Phi)$ or $s(\Phi)$, the closer $\Phi$ to $\Pi_\Phi$, in the distance measure given by the largest eigenvalue or singular value of their difference. Now, if the `ideal' channel $\Pi_\Phi$ has a large Kraus rank, the goal is to find an approximate channel $\Phi$ that is as `economical' as possible, i.e.~that has an as small as possible Kraus rank $d$. For instance, in the simplest case where $\lambda(\Phi)=|\lambda_2(\Phi)|$ and $s(\Phi)=s_2(\Phi)$, $\Pi_\Phi$ is just the channel that sends any input state on the maximally mixed state (the so-called fully randomizing or depolarizing channel), which has maximal Kraus rank $N^2$. So any approximation of $\Pi_\Phi$ with Kraus rank $d\ll N^2$ could be potentially useful in practice. More precisely, the idea is often to understand, given a class of quantum channels of interest, what is the best possible scaling between the parameters $d$ and $\lambda$, in order to try and exhibit examples of channels achieving it.

If the unital quantum channel $\Phi$ is self-adjoint, in the sense that $\Phi^*=\Phi$, then conditions (1) and (2) in Definition \ref{def:expander-general} above are equivalent, because $\lambda(\Phi)=s(\Phi)$. While if it is such that $r'(\Phi)=r(\Phi)$, then condition (2) is stronger than condition (1), because by Weyl's majorant theorem (see e.g.~\cite[Theorem~II.3.6]{Bhatia97}) $|\lambda_1(\Phi)|+\cdots+|\lambda_{r(\Phi)+1}(\Phi)| \leq s_1(\Phi)+\cdots+s_{r(\Phi)+1}(\Phi)$, so the fact that $|\lambda_i(\Phi)|=s_i(\Phi)=1$ for all $1\leq i\leq r(\Phi)$ implies that $\lambda(\Phi)=|\lambda_{r(\Phi)+1}(\Phi)| \leq s_{r(\Phi)+1}(\Phi)=s(\Phi)$. As we will see, the latter property will be satisfied by the unital quantum channels that we will consider in the sequel. In what follows, we will thus use the stronger condition (2) as our definition of expansion. We however mention that, depending on the context, either one or the other way of defining expansion might be more relevant. For instance, $\lambda(\Phi)$ quantifies the speed of convergence of the dynamics $(\Phi^q(\rho))_{q\in\mathbb N}$ to its equilibrium and the speed of decay of correlations in the 1D many-body quantum state that has $\Phi$ as so-called transfer operator (see e.g.~\cite[Section 4]{Lancien22}). On the other hand, $s(\Phi)$ generally appears when quantifying the distance of $\Phi$ to $\Pi_\Phi$ in certain norm distances: it is always upper bounded and often equal to the $2{\to}2$ norm of $\Phi-\Pi_\Phi$ (and can be related to several other $p{\to}q$ norms of $\Phi-\Pi_\Phi$).

\subsection{Unitary designs and tensor product expanders} \hfill\vspace{0.1cm}

In what follows, we denote by $\mu_H$ the Haar measure on the unitary group $U(n)$. Given $k\in\N$, we define
\begin{equation} \label{def:P^k}
	P^{(k)} := \E_{U\sim\mu_H}\left[\left(U^{\otimes k}\otimes\overbar{U}^{\otimes k}\right)\right] \in \mathcal M_{n^{2k}}(\C).
\end{equation}
Identifying linear maps on $\C^{n^k}\otimes\C^{n^k}$, i.e.~elements of $\mathcal M_{n^{2k}}(\C)$, with linear maps on $\mathcal M_{n^k}(\C)$, we see that $P^{(k)}$ is nothing else than the matrix version of the unital quantum channel
\begin{equation} \label{def:Pi^k}
	\Pi^{(k)}: Y\in\mathcal M_{n^k}(\mathbb C) \mapsto \E_{U\sim\mu_H}\left[\left(U^{\otimes k}YU^{*\otimes k}\right)\right] \in\mathcal M_{n^k}(\mathbb C) . 
\end{equation}

Let us first illustrate these definitions in the simplest cases where $k=1,2$. $P^{(1)}$ and $P^{(2)}$ can be easily expressed explicitly, as
\begin{align*} 
	P^{(1)} & = \Psi_{12}, \\ 
	P^{(2)} & = \frac{n}{2(n+1)}\left(\ket{\psi_{13}\otimes\psi_{24}} +\ket{\psi_{14}\otimes\psi_{23}}\right)\left(\bra{\psi_{13}\otimes\psi_{24}} +\bra{\psi_{14}\otimes\psi_{23}}\right) \\
	& \qquad + \frac{n}{2(n-1)}\left(\ket{\psi_{13}\otimes\psi_{24}} -\ket{\psi_{14}\otimes\psi_{23}}\right)\left(\bra{\psi_{13}\otimes\psi_{24}} -\bra{\psi_{14}\otimes\psi_{23}}\right) \\
	& = \left(1+\frac{1}{n^2-1}\right)\left(\Psi_{13}\otimes\Psi_{24}+\Psi_{14}\otimes\Psi_{23}\right) - \frac{n}{n^2-1}\left(\ketbra{\psi_{13}\otimes\psi_{24}}{\psi_{14}\otimes\psi_{23}} + \ketbra{\psi_{14}\otimes\psi_{23}}{\psi_{13}\otimes\psi_{24}} \right),
\end{align*}
where $\Psi=\ketbra{\psi}{\psi}=\sum_{i,j=1}^n\ketbra{ii}{jj}/n$ denotes the maximally entangled state on $\C^n\otimes\C^n$ (and subscripts indicate which copies of $\C^n$ each operator is acting on).
This translates into the following explicit expressions for $\Pi^{(1)}$ and $\Pi^{(2)}$
\begin{align*} 
	\forall\ Y\in\mathcal M_n(\C),\ \Pi^{(1)}(Y) & = \tr(Y)\frac{I}{n}, \\ 
	\forall\ Y\in\mathcal M_{n^2}(\C),\ \Pi^{(2)}(Y) & = \tr\left(Y\frac{I+F}{2}\right)\frac{I+F}{n(n+1)} + \tr\left(Y\frac{I-F}{2}\right)\frac{I-F}{n(n-1)}, 
\end{align*}
where $F=\sum_{i,j=1}^n\ketbra{ij}{ji}$ denotes the flip operator on $\C^n\otimes\C^n$.

More generally, we can write down $P^{(k)}$, for any $k\in\mathbb N$, thanks to Weingarten calculus (see e.g.~\cite[Section 2]{Collins06} for more details). In what follows, we denote by $\mathcal S_k$ the set of permutations of $k$ elements. And given $\pi\in\mathcal S_k$, we denote by $|\pi|$ the minimal number of transpositions needed to write $\pi$ as a product of transpositions, by $C(\pi)$ the set of cycles appearing in the decomposition of $\pi$ into disjoint cycles and by $\sharp(c)$ the length of $c\in C(\pi)$. An explicit expression for $P^{(k)}$ then reads
\begin{equation} \label{eq:P^k-Weingarten}
	P^{(k)} = \sum_{\pi,\sigma\in\mathcal S_k} \text{Wg}(n,\pi\sigma^{-1}) \left( \sum_{1\leq i_1,j_1,\ldots,i_k,j_k\leq n} \ketbra{i_1\cdots i_k i_{\pi(1)}\cdots i_{\pi(k)}}{j_1\cdots j_k j_{\sigma(1)}\cdots j_{\sigma(k)}} \right),
\end{equation}
where the Weingarten function $\text{Wg}(n,\pi)$ satisfies
\begin{equation} \label{eq:Weingarten-Moebius} \text{Wg}(n,\pi)=\frac{1}{n^{k+|\pi|}}\left(\text{Mb}(\pi)+O\left(\frac{1}{n^2}\right)\right), 
\end{equation}
and the Moebius function $\text{Mb}(\pi)$ is defined as
\[ \text{Mb}(\pi)=\prod_{c\in C(\pi)} (-1)^{\sharp(c)-1}\mathrm{Cat}_{\sharp(c)-1}, \]
for $\mathrm{Cat}_\ell=\binom{2\ell}{\ell}/(\ell+1)$ the $\ell$-th Catalan number.

$P^{(k)}$ has the crucial property of being the orthogonal projector onto the subspace 
\[ E^{(k)} := \mathrm{span}\left\{ \ket{u_\pi} := \sum_{1\leq i_1,\ldots,i_k\leq n}\ket{i_1\cdots i_k i_{\pi(1)}\cdots i_{\pi(k)}}  : \pi\in\mathcal S_k \right\} \subset \C^{n^{2k}}, \]
whose dimension we denote by $r(n,k)$. $P^{(k)}$ thus has only $1$ and $0$ as eigenvalues and singular values, with multiplicities $r(n,k)$ and $n^{2k}-r(n,k)$ respectively. Note that, if $k\leq n$, then the vectors $u_\pi$, $\pi\in\mathcal S_k$, are linearly independent, and we thus simply have $r(n,k)=k!$. In what follows, this will always be satisfied, since we will consider the regime where $k$ is fixed and $n$ grows.

A last straightforward but key property of $P^{(k)}$ that we should mention is that it is invariant under left and right multiplication by a unitary of the form $V^{\otimes k}\otimes\overbar{V}^{\otimes k}$, i.e.
\[ \forall\ V\in U(n),\ V^{\otimes k}\otimes\overbar{V}^{\otimes k}P^{(k)} = P^{(k)}V^{\otimes k}\otimes\overbar{V}^{\otimes k} = P^{(k)}. \]
This is a consequence of the fact that the Haar measure $\mu_H$ on $U(n)$ is itself (by definition) left and right invariant.

We are now ready to get to the two main definitions of this work, following the terminology adopted in \cite{Harrow09}. For that, we need to introduce one final notation. Given $k\in\N$, we define for any probability measure $\mu$ on the unitary group $U(n)$
\[ P^{(k)}_\mu:= \E_{U\sim\mu}\left[\left(U^{\otimes k}\otimes\overbar{U}^{\otimes k}\right)\right] \in\mathcal M_{n^{2k}}(\C), \]
which is the matrix version of the unital quantum channel
\[ \Pi^{(k)}_\mu: Y\in\mathcal M_{n^k}(\mathbb C) \mapsto \E_{U\sim\mu}\left[\left(U^{\otimes k}YU^{*\otimes k}\right)\right] \in\mathcal M_{n^k}(\mathbb C). \]
Those generalize the definitions of $P^{(k)}$ and $\Pi^{(k)}$ given in equations \eqref{def:P^k} and \eqref{def:Pi^k}, corresponding to the special case where $\mu=\mu_H$.

\begin{definition}[Unitary design] \label{def:design}
	Given $k\in\mathbb N$, a probability measure $\mu$ on the unitary group $U(n)$ is called a (unitary) $k$-design if 
	\[ P^{(k)}_\mu = P^{(k)} \quad \text{i.e.~equivalently} \quad \Pi^{(k)}_\mu =  \Pi^{(k)}. \]
\end{definition}

Note that, if $\mu$ is a $k$-design for some $k\in\N$, then it is automatically a $k'$-design for all $k'\leq k$. So satisfying the design property up to some order $k$ should really be seen as being an approximation of the Haar measure, with precision increasing with $k$.

\begin{definition}[Tensor product expander] \label{def:expander}
	Fix $0<\lambda<1$. Given $k\in\mathbb N$, a probability measure $\mu$ on the unitary group $U(n)$ whose support has finite cardinality $d=|\mathrm{supp}(\mu)|$ is called a $(d,\lambda)$ $k$-copy tensor product expander if 
	\[ \left\| P^{(k)}_\mu - P^{(k)} \right\|_\infty \leq \lambda \quad \text{i.e.~equivalently} \quad \left\| \Pi^{(k)}_\mu - \Pi^{(k)} \right\|_\infty \leq \lambda. \]
\end{definition}

Note that, for any probability measure $\mu$ on the unitary group $U(n)$, the subspace $E^{(k)}$ is (left and right) invariant under the action of the matrix $P^{(k)}_\mu$. So the largest eigenvalue and singular value $1$ of $P^{(k)}_\mu$ have multiplicity at least $r(n,k)$, with corresponding Jordan block the orthogonal projection $P^{(k)}$ onto $E^{(k)}$. What is more, if $\mu$ is a $(d,\lambda)$ $k$-copy tensor product expander for some $0<\lambda<1$, then it implies that $P^{(k)}_\mu$ has no other eigenvalue of modulus equal to $1$ or singular value equal to $1$ than those of $P^{(k)}$. Hence in this case, using the notations introduced in equations \eqref{eq:def-lambda} and \eqref{eq:def-s}, we have 
\[ r\left(\Pi^{(k)}_\mu\right)=r'\left(\Pi^{(k)}_\mu\right)=r(n,k) \quad \text{and} \quad \lambda\left(\Pi^{(k)}_\mu\right)=\left|\lambda_1\left(\Pi^{(k)}_\mu - \Pi^{(k)}\right)\right|,\ s\left(\Pi^{(k)}_\mu\right)=s_1\left(\Pi^{(k)}_\mu - \Pi^{(k)}\right) . \]
We thus see that the characterization of an expander given in Definition \ref{def:expander}, which applies to the specific case of unital quantum channels whose Kraus operators are (weighted) tensor products of unitaries, actually coincides with the general one given in Definition \ref{def:expander-general}. We have just made the slight abuse of saying indifferently that a probability measure $\mu$ on the unitary group or its corresponding unital quantum channel $\Pi_\mu^{(k)}$ is an expander.

\subsection{Presentation of the problem and main results} \hfill\vspace{0.1cm}

The general question that we address in this work is whether optimal (tensor product) expanders can be constructed by sampling Kraus operators from unitary designs. Let us first explain what we mean by optimal in this context. For that, we need the following lemma, which gives us a lower bound on how expanding a quantum channel that is a mixture of unitary conjugations can be, as the underlying dimension increases. 

\begin{lemma} \label{lem:s_2-lb}
	Let $d=d_N,r=r_N\in\mathbb N$ be such that $d=o(N^2/r)$ as $N\to\infty$. Given $V_1,\ldots,V_d\in U(N)$, define the quantum channel 
	\[ \Phi:X\in\mathcal M_N(\mathbb C) \mapsto \frac{1}{d}\sum_{s=1}^dV_sXV_s^*\in\mathcal M_N(\mathbb C), \]
	and suppose that $r'(\Phi)=r$, where $r'(\Phi)$ is as defined in equation \eqref{eq:def-s}. Then, as $N\to\infty$, we have
	\begin{equation} \label{eq:s_2-lb}
		s(\Phi) \geq \frac{2\sqrt{d-1}}{d}\left(1-o(1)\right),
	\end{equation}
	where $s(\Phi)$ is also as defined in equation \eqref{eq:def-s}.
\end{lemma}

This result essentially appears as \cite[Lemma 1.5]{Timhadjelt24}. We simply extend it here to the case where the largest singular value might be degenerate. This generalization is straightforward so we repeat only the main steps of the proof, for the sake of completeness (and we refer the reader to it for further details).

\begin{proof}
	Let $p\in\mathbb N$, to be determined later, and set
	\[ m^{(p)}(\Phi) := \sum_{i=1}^{N^2} s_i(\Phi)^{2p} = \frac{1}{d^{2p}}\sum_{s_1,\ldots,s_{2p}=1}^d \left|\tr\left(V_{s_1}V^*_{s_2}\cdots V_{s_{2p-1}}V^*_{s_{2p}}\right)\right|^2. \]
	The last equality is because, defining $M_\Phi=\sum_{s=1}^d V_s\otimes\overbar{V}_s/d$ as the matrix version of $\Phi$, we have 
	\begin{align*} 
		\sum_{i=1}^{N^2} s_i(\Phi)^{2p} & = \tr\left(\left(M_\Phi M_\Phi^*\right)^p\right) \\
		& = \tr\left( \frac{1}{d^{2p}}\sum_{s_1,\ldots,s_{2p}=1}^d \left(V_{s_1}V^*_{s_2}\cdots V_{s_{2p-1}}V^*_{s_{2p}}\right)\otimes\left(\overbar{V}_{s_1}\overbar{V}^*_{s_2}\cdots \overbar{V}_{s_{2p-1}}\overbar{V}^*_{s_{2p}}\right) \right) \\
		& = \frac{1}{d^{2p}}\sum_{s_1,\ldots,s_{2p}=1}^d \tr\left(V_{s_1}V^*_{s_2}\cdots V_{s_{2p-1}}V^*_{s_{2p}}\right) \times \tr\left(\overbar{V}_{s_1}\overbar{V}^*_{s_2}\cdots \overbar{V}_{s_{2p-1}}\overbar{V}^*_{s_{2p}}\right) .
	\end{align*}
	Now on the one hand, we have by definition
	\[ m^{(p)}(\Phi) \leq r + (N^2-r)\times s(\Phi) \leq r + N^2\times s(\Phi). \]
	And on the other hand, setting
	\[ \eta(p,d) := \left|\left\{ (s_1,\ldots,s_{2p})\in\{1,\ldots,d\}^{2p} : V_{s_1}V^*_{s_2}\cdots V_{s_{2p-1}}V^*_{s_{2p}}=I \right\}\right|, \]
	we clearly have
	\[ m^{(p)}(\Phi) \geq N^2\times\frac{\eta(p,d)}{d^{2p}}. \]
	Since $\eta(p,d)\geq c(2\sqrt{d-1})^{2p}/p^{3/2}$ for some absolute constant $c>0$ (see e.g.~\cite[Proof of Theorem 5.3]{HLW06}), we deduce from the above two inequalities that
	\[ s(\Phi) \geq \frac{2\sqrt{d-1}}{d}\left(\frac{c}{p^{3/2}}-\frac{r}{N^2}\left(\frac{d}{2\sqrt{d-1}}\right)^{2p}\right)^{1/2p}. \]
	Now, choosing $p=p_N$ such that, as $N\to\infty$, 
	\[ p= o\left(\frac{\log\left(N^2/r\right)}{2\log\left(d/2\sqrt{d-1}\right)}\right) \quad \text{and} \quad p\to\infty, \]
	which is indeed possible with the assumption that $d=o(N^2/r)$,
	we have that, as $N\to\infty$,
	\[ \left(\frac{c}{p^{3/2}}-\frac{r}{N^2}\left(\frac{d}{2\sqrt{d-1}}\right)^{2p}\right)^{1/2p} \to 1. \]
	And this establishes the announced result.
\end{proof}

In what follows, we will apply Lemma \ref{lem:s_2-lb} to the situation where $r=k!$ is fixed while both $N=n^k$ and $d$ are large (under the constraint that $d$ is small compared to $N^2=n^{2k}$). In this regime, we will say that a quantum channel $\Phi$ on $\mathcal M_N(\C)$ that is a mixture of unitary conjugations is an optimal expander if it (asymptotically) achieves the lower bound \eqref{eq:s_2-lb}, i.e.~if
\begin{equation} \label{def:optimal-expander}
	s(\Phi) \leq  \frac{2}{\sqrt{d}}(1+\delta_{N,d}) , 
\end{equation}
with $\delta_{N,d}\to 0$ as $N,d\to\infty$ (with $d\ll N^2$).

We first investigate the simplest setup, where Kraus operators are of the form $K_s=U_s/\sqrt{d}$ with $U_s\in U(n)$, $1\leq s\leq d$. It was originally shown by Hastings in \cite{Hastings07} that, sampling the $U_s$'s as $d/2$ independent Haar distributed unitaries, together with their adjoints, provided with high probability an optimal expander, in the sense of equation \eqref{def:optimal-expander}. Pisier then proved in \cite{Pisier14} that the same is true for the $U_s$'s being $d$ independent Haar distributed unitaries, up to a constant multiplicative factor (i.e.~replacing $2/\sqrt{d}$ by $C/\sqrt{d}$ for some absolute constant $2\leq C<\infty$ in equation \eqref{def:optimal-expander}). This was very recently improved by Timhadjelt in \cite{Timhadjelt24}, who established this result with the optimal constant $2$. It is also worth pointing out that the latter work is the first one that addresses the case where $d$ is growing with $n$: all previous ones only covered the regime of fixed $d$ and growing $n$.

The main problem with all these random constructions is that sampling large Haar distributed unitaries requires a large amount of randomness, and is thus costly to implement in practice. One can therefore naturally wonder whether random examples of optimal expanders can be obtained by sampling Kraus operators on the unitary group from a simpler measure than the Haar measure, for instance one that is finite rather than continuous. Unitary designs are natural candidates for this purpose, since on the one hand they, by definition, resemble the Haar measure (up to moments of some order) and on the other they can be efficiently generated. Indeed, there are explicit finite subsets of unitaries that are known to form exact designs, and it was proven through a long line of works that random quantum circuits of short depth are with high probability approximate designs (see Section \ref{sec:summary} for a more in-depth discussion on this point).

In this work, we prove that, in the regime $(\log n)^4\ll d\ll n^2$, a random unital quantum channel on $\mathcal M_n(\C)$ whose Kraus operators are sampled as $d$ independent unitaries from a $2$-design is with high probability an optimal expander, in the sense of equation \eqref{def:optimal-expander}. This result appears as Corollary \ref{th:main-implication}. This is particularly interesting in practice since there are explicit examples of finite subsets of unitaries that are known to form a $2$-design. Some of them, such as e.g.~the Clifford group, are even efficiently implementable.

We then investigate the tensor power setup, where Kraus operators are of the form $K_s=U_s^{\otimes k}/\sqrt{d}$ with $U_s\in U(n)$, $1\leq s\leq d$. It was shown by Harrow and Hastings in \cite{Hastings09} that, for $d$ fixed, sampling the $U_s$'s as $d/2$ independent Haar distributed unitaries, together with their adjoints, provided with high probability a close to optimal expander, in the sense that equation \eqref{def:optimal-expander} holds with $2/\sqrt{d}$ replaced by $C/\sqrt{d}$ for some absolute constant $2\leq C<\infty$. 

Very recently, Fukuda proved in \cite{Fukuda24} that, for $d\geq C'k\log n$, sampling the $U_s$'s as $d$ independent $k$-design distributed unitaries provided with high probability an expander such that, in equation \eqref{def:optimal-expander}, $2/\sqrt{d}$ is replaced by $C\sqrt{k\log n}/\sqrt{d}$. The latter work was carried out independently of this one, and uses different techniques (namely a matrix version of Bernstein's inequality). It also contains several corollaries on approximating $\Pi^{(k)}$ in various $1{\to}p$ norms (i.e.~minimizing the Shatten $p$-norm of the output over all Schatten $1$-norm inputs). While here we only focus on approximating $\Pi^{(k)}$ in $2{\to}2$ norm, which is the relevant one in the context of expansion.

We prove in what follows that, again in the regime $(\log n)^4\ll d\ll n^{2k}$, a random unital quantum channel on $\mathcal M_{n^k}(\C)$ whose Kraus operators are sampled as $d$ independent $k$-copy tensor power unitaries from a $2k$-design is with high probability an optimal $k$-copy tensor product expander, in the sense of equation \eqref{def:optimal-expander}. This result appears as Corollary \ref{th:main-implication-k}. Exactly as for the case $k=1$, this is interesting in practice since there are explicit constructions of finite $2k$-designs \cite{Bannai22}. Note that, compared to the result of Fukuda mentioned above, our result has the advantage of establishing optimal expansion (without an extra $\sqrt{k\log n}$ factor), but at the cost of doubling the order of the design that Kraus operators are sampled from ($2k$ vs $k$) and of restricting slightly more the range of admissible $d$ (to $d\geq\mathrm{poly}(\log n)$ rather than $d\geq k\log n$).

Finally, let us say a few words about a setup that is slightly different from the one we consider here, namely the one where $d$ is fixed and $n$ is infinite. Such asymptotic regime can be tackled thanks to powerful tools from free probability. It was for instance shown by Bordenave and Collins \cite{BC24} that, if $U_1,\ldots,U_d$ are  independent Haar distributed on $U(n)$ or $O(n)$, then the quantum channel $\Phi$ that has the $U_s^{\otimes k}/\sqrt{d}$ as Kraus operators is such that $s(\Phi)$ converges almost surely to $2\sqrt{d-1}/d$ as $n\to\infty$. This is a consequence of the fact that the $U_s^{\otimes k}\otimes\overbar{U}_s^{\otimes k}$ are asymptotically strongly free on the subspace orthogonal to their fixed subspace, which allows to identify the asymptotic almost sure operator norm of $\sum_{s=1}^dU_s^{\otimes k}\otimes\overbar{U}_s^{\otimes k}$ when restricted to this subspace.  

\subsection{Proof strategy and main technical tool} \hfill\vspace{0.1cm}

The model of random unital quantum channel that we consider here is of the form
\[ \Phi:Y\in\mathcal M_{n^k}(\C)\mapsto\frac{1}{d}\sum_{s=1}^d U_s^{\otimes k}YU_s^{*\otimes k}\in\mathcal M_{n^k}(\C), \]
where $U_1,\ldots,U_d\in U(n)$ are independently sampled from a $2k$-design. The goal is to prove that there exist $\delta_n,\epsilon_n\to_{n\to\infty}0$ such that
\[ \P\left( \left\|\Phi-\Pi^{(k)}\right\|_\infty \leq \frac{2}{\sqrt{d}}\left(1+\delta_n\right) \right) \geq 1-\epsilon_n, \]
hence establishing that $\Phi$ is typically an optimal $k$-copy tensor product expander. As already explained, this equivalent to proving that 
\[ \P\left( \left\|M_\Phi-P^{(k)}\right\|_\infty \leq \frac{2}{\sqrt{d}}\left(1+\delta_n\right) \right) \geq 1-\epsilon_n, \]
where $M_\Phi\in\mathcal M_{n^{2k}}(\C)$ is the matrix version of $\Phi$, i.e.
\[ M_\Phi = \frac{1}{d}\sum_{s=1}^d U_s^{\otimes k}\otimes\overbar{U}_s^{\otimes k}. \]

A straightforward but key observation is that $\Pi^{(k)}=\E(\Phi)$, or equivalently $P^{(k)}=\E(M_\Phi)$. This is because the $U_s$'s being drawn from a $2k$-design, they are in particular drawn from a $k$-design. So what we ultimately have to upper bound is
\[ \left\|\Phi-\E(\Phi)\right\|_\infty = \left\|M_\Phi-\E(M_\Phi)\right\|_\infty. \]
The final results that we obtain appear as Theorem \ref{th:main} for the particular case $k=1$ and as Theorem \ref{th:main-k} for the general case.

In order to do this, we will follow the proof strategy adopted in \cite{Lancien23}, which studied a similar question for random models of non-unital quantum channels. Concretely, as it was done in the latter work, we will make use of recent advances in the study of the operator norm of random matrices with dependence and non-homogeneity, which culminated with the works \cite{vanHandel21,vanHandel22}. More precisely, we will rely on the following result which appears as \cite[Corollary 2.17]{vanHandel22}. We state here only the version that will be useful to us, namely for (non-self-adjoint complex) matrices that are almost surely bounded, but there are numerous variations suited to slightly different settings.

\begin{theorem}[Operator norm of random matrices with dependence and non-homogeneity] \label{th:vanHandel}
	Let $Z_1,\ldots,Z_d\in\mathcal M_N(\C)$ be independent centered and almost surely bounded random matrices, and set $X=\sum_{s=1}^d Z_s$. Let $\Cov(X)$ denote the covariance matrix associated to $X$, i.e.~the $N^2\times N^2$ matrix such that $\Cov(X)_{ijkl}= \E(X_{ij}\overbar{X}_{kl})$ for every $1\leq i,j,k,l\leq N$. Define the following three parameters of the random matrix $X$:
	\begin{align*}
		& \sigma(X):= \max\left(\|\E(XX^*)\|_\infty^{1/2},\|\E(X^*X)\|_\infty^{1/2}\right), \\
		& \upsilon(X):= \|\Cov(X)\|_\infty^{1/2}, \\
		& R(X) := \max_{1\leq s\leq d}\left\|Z_s\right\|_\infty.
	\end{align*}
	We then have 
	\begin{align*} 
		\E \|X\|_\infty & \leq \|\E(XX^*)\|_\infty^{1/2}+\|\E(X^*X)\|_\infty^{1/2} \\
		& \qquad + C\left( (\log N)^{3/4} \sigma(X)^{1/2}\upsilon(X)^{1/2} + (\log N)^{2/3} \sigma(X)^{2/3}R(X)^{1/3} + (\log N)R(X) \right), 
	\end{align*}
	where $C<\infty$ is an absolute constant. What is more, for all $t>0$, with probability at least $1-4Ne^{-t}$,
	\begin{align*} 
		\|X\|_\infty & \leq \|\E(XX^*)\|_\infty^{1/2}+\|\E(X^*X)\|_\infty^{1/2} \\
		& \qquad + C'\left( (\log N)^{3/4} \sigma(X)^{1/2}\upsilon(X)^{1/2} + \sigma(X)^{2/3}R(X)^{1/3}t^{2/3} + R(X)t + \upsilon(X)t^{1/2} \right),
	\end{align*}
	where $C'<\infty$ is an absolute constant.
\end{theorem}
In \cite[Corollary 2.17]{vanHandel22}, a fourth parameter of the random matrix $X$ actually appears, denoted $\sigma_*(X)$. As explained in \cite[Remark 2.4]{vanHandel22}, this parameter $\sigma_*(X)$ is always upper bounded by the parameter $\upsilon(X)$. In order to make the presentation slightly lighter, we have thus substituted each occurrence of $\sigma_*(X)$ by $\upsilon(X)$ in the above reformulation of \cite[Corollary 2.17]{vanHandel22}.

Before moving on to proving our main results, we introduce two additional notations that we will use in our proofs. Namely, given a matrix $M$ on $\C^N\otimes\C^N$, of the form
\[ M=\sum_{i,j,k,l=1}^N M_{ijkl}\ketbra{ij}{kl}, \] 
we denote by $M^R$ and $M^S$ its realignment and swap, respectively, which are the matrices on $\C^N\otimes\C^N$, defined as
\[ M^R=\sum_{i,j,k,l=1}^N M_{ijkl}\ketbra{ik}{jl} \quad \text{and} \quad M^S=\sum_{i,j,k,l=1}^N M_{ijkl}\ketbra{ji}{lk}. \] 
Also note by the way that, for any $N\times N$ random matrix $Y$, we can write the $N^2\times N^2$ matrix $\Cov(Y)$ as
\[ \Cov(Y) = \E\left(Y\otimes\overbar{Y}\right)^R. \]
This technical observation will be useful to us later on. 

\section{Optimal expanders from unitary \texorpdfstring{$2$}{2}-designs}
\label{sec:1-copy}

In this section, we treat the case where $k=1$. Namely, we prove that a random unital quantum channel on $\mathcal M_n(\C)$ whose Kraus operators are sampled as $d$ independent unitaries from a $2$-design is with high probability an optimal expander, in the sense of equation \eqref{def:optimal-expander}. Indeed, in this particular case, all computations can be done explicitly, without resorting to more abstract Weingarten calculus arguments. It is thus quite instructive to deal with it separately, even though the obtained results can also be seen, a posteriori, as a consequence of those we will derive in Section \ref{sec:k-copy} for any fixed integer $k$. What is more, one of the intermediate facts that we prove in the $k=1$ case (namely Lemma \ref{lem:upsilon}) is tight, contrary to its general $k\in\N$ counterpart (namely Lemma \ref{lem:upsilon-k}). It might thus be worth having it stated separately.

\subsection{Main technical results} \hfill\vspace{0.1cm}

Let $\mu$ be a $2$-design on $U(n)$ and let $U\in U(n)$ be sampled from $\mu$. Take $U_1,\ldots,U_d\in U(n)$ independent copies of $U$ and set
\begin{equation} \label{eq:def-X} X:= \frac{1}{d}\sum_{s=1}^d \left(U_s\otimes\overbar{U}_s-\E(U_s\otimes\overbar{U}_s)\right) = \frac{1}{d}\sum_{s=1}^d \left(U_s\otimes\overbar{U}_s-\Psi\right), \end{equation}
where the last equality is because $\mu$ is in particular a $1$-design, and hence, for each $1\leq s\leq d$, $\E(U_s\otimes\overbar{U}_s)=P^{(1)}=\Psi$. 

We begin with estimating the parameters $\sigma(X)$, $\upsilon(X)$ and $R(X)$ appearing in Theorem \ref{th:vanHandel} for the random matrix $X$ given by equation \eqref{eq:def-X}. This is the content of the following Lemmas \ref{lem:sigma}, \ref{lem:upsilon} and \ref{lem:Rp} respectively.

\begin{lemma} \label{lem:sigma}
	Let $X$ be a random matrix on $\mathbb C^n\otimes\mathbb C^n$ defined as in equation \eqref{eq:def-X}. Then,
	\[ \left\|\E(XX^*)\right\|_\infty = \left\|\E(X^*X)\right\|_\infty = \frac{1}{d}. \]
\end{lemma}

\begin{proof}
	Observe that
	\begin{align*} 
		\E(XX^*) & = \frac{1}{d^2}\sum_{s,t=1}^d \E\left((U_s\otimes\overbar{U}_s-\Psi) (U_t\otimes\overbar{U}_t-\Psi)^*\right) \\
		& = \frac{1}{d^2}\sum_{s=1}^d \E\left((U_s\otimes\overbar{U}_s-\Psi) (U_s\otimes\overbar{U}_s-\Psi)^*\right) \\
		& = \frac{1}{d}\E\left((U\otimes\overbar{U}-\Psi) (U\otimes\overbar{U}-\Psi)^*\right), 
	\end{align*}
	where the second equality is because, for $1\leq s\neq t\leq d$, by independence of $U_s$ and $U_t$,
	\[ \E\left((U_s\otimes\overbar{U}_s-\Psi) (U_t\otimes\overbar{U}_t-\Psi)^*\right) = \E\left(U_s\otimes\overbar{U}_s-\Psi\right) \E\left(U_t\otimes\overbar{U}_t-\Psi\right)^* = 0. \]
	Now, we have
	\[ \E\left((U\otimes\overbar{U}-\Psi) (U\otimes\overbar{U}-\Psi)^*\right) = \E\left((U\otimes\overbar{U}) (U\otimes\overbar{U})^*\right)-\E\left(U\otimes\overbar{U}\right)\Psi - \Psi\E\left(U\otimes\overbar{U}\right)^* + \Psi = I-\Psi, \] 
	where the last equality is because $\E(U\otimes\overbar{U})=\E(U\otimes\overbar{U})^*=\Psi$, while $(U\otimes\overbar{U}) (U\otimes\overbar{U})^*=I$. Hence,
	\[ \E(XX^*) = \frac{1}{d}(I-\Psi), \]
	and we indeed have 
	\[ \left\|\E(XX^*)\right\|_\infty = \frac{1}{d}. \]
	The argument is exactly the same for $\|\E(X^*X)\|_\infty$.
\end{proof}

\begin{lemma} \label{lem:upsilon}
	Let $X$ be a random matrix on $\mathbb C^n\otimes\mathbb C^n$ defined as in equation \eqref{eq:def-X}. Then,
	\[ \left\|\Cov(X)\right\|_\infty = \frac{1}{d}\times\frac{1}{n^2-1}. \]
\end{lemma}

\begin{proof}
	Observe that, denoting by $R$ the realignment operation between tensor factors $\{1,2\}$ and $\{3,4\}$,
	\begin{align*} 
		\Cov(X) & = \E\left(X\otimes\overbar{X}\right)^R \\
		& = \frac{1}{d^2} \sum_{s,t=1}^d \E\left((U_s\otimes\overbar{U}_s-\Psi)\otimes (\overbar{U}_t\otimes U_t-\Psi)\right)^R \\
		& = \frac{1}{d^2} \sum_{s=1}^d \E\left((U_s\otimes\overbar{U}_s-\Psi)\otimes (\overbar{U}_s\otimes U_s-\Psi)\right)^R \\
		& = \frac{1}{d} \E\left((U\otimes\overbar{U}-\Psi)\otimes (\overbar{U}\otimes U-\Psi)\right)^R,
	\end{align*}
	where the third equality is because, for $1\leq s\neq t\leq d$, by independence of $U_s$ and $U_t$,
	\[ \E\left((U_s\otimes\overbar{U}_s-\Psi)\otimes (\overbar{U}_t\otimes U_t-\Psi)\right)^R = \left(\E\left(U_s\otimes\overbar{U}_s-\Psi\right)\otimes \E\left(\overbar{U}_t\otimes U_t-\Psi\right)\right)^R = 0. \]
	Now, we have
	\begin{align*} 
		\E\left((U\otimes\overbar{U}-\Psi)\otimes (\overbar{U}\otimes U-\Psi)\right) & = \E\left(U\otimes\overbar{U}\otimes\overbar{U}\otimes U\right)-\E\left(U\otimes\overbar{U}\right)\otimes\Psi - \Psi\otimes\E\left(\overbar{U}\otimes U\right) + \Psi\otimes\Psi \\
		& = \E\left(U\otimes\overbar{U}\otimes\overbar{U}\otimes U\right)-\Psi\otimes\Psi, 
	\end{align*}
	where the last equality is because $\E(U\otimes\overbar{U})=\E(\overbar{U}\otimes U)=\Psi$. Next, denoting by $S$ the swap operation between tensor factors $2$ and $4$ (leaving tensor factors $1$ and $3$ unchanged), we have 
	\begin{align*}
		\E\left(U\otimes\overbar{U}\otimes\overbar{U}\otimes U\right) & = \E\left(U\otimes U\otimes\overbar{U}\otimes\overbar{U}\right)^{S} \\
		& = \left(1+\frac{1}{n^2-1}\right)\left(\Psi_{13}\otimes\Psi_{24}+\Psi_{12}\otimes\Psi_{34}\right) \\
		& \qquad - \frac{n}{n^2-1}\left(\ketbra{\psi_{13}\otimes\psi_{24}}{\psi_{12}\otimes\psi_{34}} + \ketbra{\psi_{12}\otimes\psi_{34}}{\psi_{13}\otimes\psi_{24}} \right).
	\end{align*}
	And therefore,
	\[ \E\left(U\otimes\overbar{U}\otimes\overbar{U}\otimes U\right)^R = \left(1+\frac{1}{n^2-1}\right)\Psi_{12}\otimes\Psi_{34} + \frac{1}{n^2-1}\left(I_{1234}-I_{12}\otimes\Psi_{34}-\Psi_{12}\otimes I_{34} \right). \]
	Indeed, it is easy to check that $(\Psi_{12}\otimes\Psi_{34})^R=\Psi_{12}\otimes\Psi_{34}$ and $(\Psi_{13}\otimes\Psi_{24})^R=I_{1234}/n^2$, while $(\ketbra{\psi_{13}\otimes\psi_{24}}{\psi_{12}\otimes\psi_{34}})^R=I_{12}\otimes\Psi_{34}/n$ and $(\ketbra{\psi_{12}\otimes\psi_{34}}{\psi_{13}\otimes\psi_{24}})^R=\Psi_{12}\otimes I_{34}/n$.
	Hence, putting everything together, and using once again that $(\Psi_{12}\otimes\Psi_{34})^R=\Psi_{12}\otimes\Psi_{34}$, we get
	\[ \left(\E\left(U\otimes\overbar{U}\otimes\overbar{U}\otimes U\right)-\Psi\otimes\Psi\right)^R = \frac{1}{n^2-1}\left(\Psi\otimes\Psi+I\otimes I-I\otimes\Psi-\Psi\otimes I\right) = \frac{1}{n^2-1}\left(I-\Psi\right)\otimes\left(I-\Psi\right), \]
	so that
	\[ \Cov(X) = \frac{1}{d}\times\frac{1}{n^2-1}\left(I-\Psi\right)\otimes\left(I-\Psi\right), \]
	and we indeed have
	\[ \left\|\Cov(X)\right\|_\infty = \frac{1}{d}\times\frac{1}{n^2-1}. \]
\end{proof}

\begin{lemma} \label{lem:Rp}
	Let $X=\sum_{s=1}^d Z_s$ be a random matrix on $\mathbb C^n\otimes\mathbb C^n$ defined as in equation \eqref{eq:def-X}, where $Z_s=(U_s\otimes\overbar{U}_s-\Psi)/d$, $1\leq s\leq d$. Then,
	\[ \max_{1\leq s\leq d}\left\|Z_s\right\|_\infty = \frac{1}{d}. \]
\end{lemma}

\begin{proof}
	Observe that
	\[ \max_{1\leq s\leq d}\left\|Z_s\right\|_\infty = \frac{1}{d}\max_{1\leq s\leq d}\left\|U_s\otimes\overbar{U}_s-\Psi\right\|_\infty = \frac{1}{d}\left\|U\otimes\overbar{U}-\Psi\right\|_\infty. \]
	Now, we have
	\begin{align*} 
		\left\|U\otimes\overbar{U}-\Psi\right\|_\infty & = \left\|\left(U\otimes\overbar{U}-\Psi\right)\left(U\otimes\overbar{U}-\Psi\right)^*\right\|_\infty^{1/2} \\
		& = \left\|(UU^*)\otimes(\overbar{U}\overbar{U}^*)-(U\otimes\overbar{U})\Psi-\Psi(U^*\otimes\overbar{U}^*)+\Psi\right\|_\infty^{1/2} \\
		& = \left\|I-\Psi\right\|_\infty^{1/2} \\
		& = 1, 
	\end{align*}
	where the third equality is because $U\otimes\overbar{U}\ket{\psi}=\ket{\psi}$ and $UU^*=\overbar{U}\overbar{U}^*=I$. So we indeed have
	\[ \max_{1\leq s\leq d}\left\|Z_s\right\|_\infty = \frac{1}{d}. \]
\end{proof}

We point out that, for the results of Lemmas \ref{lem:sigma} and \ref{lem:Rp} to hold, it is actually enough that the operators $U_1,\ldots,U_d\in U(n)$ are independent unitaries, without any extra assumption on their distribution. It is only for Lemma \ref{lem:upsilon} that we need them to be sampled from a $2$-design. Let us try and briefly explain why this is so. Setting $Z=U\otimes\overbar U-\Psi$, we have already explained that, simply because $U\in U(n)$, we have 
\[ ZZ^* = (UU^*)\otimes(\overbar{U}\overbar{U}^*)-(U\otimes\overbar{U})\Psi-\Psi(U^*\otimes\overbar{U}^*)+\Psi = I-\Psi. \]
And therefore $\|ZZ^*\|_\infty=1$. 
On the other hand, we have
\[ \Cov(Z) = \E\left(Z\otimes\overbar Z\right)^R = \left( \E\left(U\otimes\overbar U\otimes\overbar U\otimes U\right) - \Psi \right)^R. \]
So in order to control $\|\Cov(Z)\|_\infty$ we need to have information on $\E(U^{\otimes 2}\otimes\overbar U^{\otimes 2})$.

With Lemmas \ref{lem:sigma}, \ref{lem:upsilon} and \ref{lem:Rp} at hand, we are now ready to state the main result of this section.

\begin{theorem} \label{th:main}
	Let $X$ be a random matrix on $\mathbb C^n\otimes\mathbb C^n$ defined as in equation \eqref{eq:def-X}. Suppose that $d\geq(\log n)^{4+\epsilon}$ for some $\epsilon>0$. Then,
	\[ \E\|X\|_\infty \leq \frac{2}{\sqrt{d}}\left(1+\frac{C}{(\log n)^{\epsilon/6}}\right), \]
	where $C<\infty$ is an absolute constant. What is more, 
	\[ \P\left( \|X\|_\infty \leq \frac{2}{\sqrt{d}}\left(1+\frac{C'}{(\log n)^{\epsilon/6}}\right) \right) \geq 1-\frac{1}{n}, \]
	where $C'<\infty$ is an absolute constant.
\end{theorem}

\begin{proof}
	By Lemmas \ref{lem:sigma}, \ref{lem:upsilon} and \ref{lem:Rp}, respectively, we can estimate the parameters $\sigma(X)$, $\upsilon(X)$ and $R(X)$ defined in Theorem \ref{th:vanHandel}. We have
	\[ \sigma(X)=\frac{1}{\sqrt{d}}, \quad \upsilon(X)=\frac{1}{\sqrt{d}}\times\frac{1}{\sqrt{n^2-1}}, \quad R(X)=\frac{1}{d}. \]
	Applying the first statement in Theorem \ref{th:vanHandel}, we thus get
	\[ \E\|X\|_\infty \leq \frac{2}{\sqrt{d}} + \frac{C}{\sqrt{d}}\left( \frac{(\log n)^{3/4}}{n^{1/2}} + \frac{(\log n)^{2/3}}{d^{1/6}} + \frac{\log n}{d^{1/2}} \right). \]
	Hence, for $d\geq(\log n)^{4+\epsilon}$,
	\[ \E\|X\|_\infty \leq \frac{2}{\sqrt{d}} + \frac{C}{\sqrt{d}}\times\frac{3}{(\log n)^{\epsilon/6}}, \]
	which is exactly the first statement in Theorem \ref{th:main} (up to relabeling $3C/2$ into $C$).
	
	To prove the second statement in Theorem \ref{th:main}, we just have to apply the second statement in Theorem \ref{th:vanHandel}, to obtain
	\[ \forall\ t>0,\ \P\left( \|X\|_\infty \leq \frac{2}{\sqrt{d}} + \frac{C}{\sqrt{d}}\left( \frac{(\log n)^{3/4}}{n^{1/2}} + \frac{t^{2/3}}{d^{1/6}} + \frac{t}{d^{1/2}} + \frac{t^{1/2}}{n} \right) \right) \geq 1-4n^2e^{-t}. \]
	Applying the above to $t=5\log n$, we have that, for $d\geq(\log n)^{4+\epsilon}$,
	\[ \P\left( \|X\|_\infty \leq \frac{2}{\sqrt{d}} + \frac{C}{\sqrt{d}}\times\frac{4\times 5}{(\log n)^{\epsilon/6}} \right) \geq 1-\frac{1}{n}, \]
	which implies exactly the announced result (setting $C'=10C$).
\end{proof}

\subsection{Implication concerning random `mixture of unitaries' quantum channels} \hfill\vspace{0.1cm}

Let $\mu$ be a $2$-design on $U(n)$ and let $U_1,\ldots,U_d\in U(n)$ be independently sampled from $\mu$. Define the random CP map
\begin{equation} \label{eq:def-Phi}
\Phi:Y\in\mathcal M_n(\mathbb C)\mapsto \frac{1}{d}\sum_{s=1}^d U_sYU_s^*\in\mathcal M_n(\mathbb C).
\end{equation}
$\Phi$ is by construction TP and unital (because $U_s^*U_s=U_sU_s^*=I$ for each $1\leq s\leq d$). Additionally, since $\mu$ is in particular a $1$-design, $\E(\Phi)=\Pi^{(1)}=\Pi$.

\begin{corollary} \label{th:main-implication}
	Let $\Phi$ be a random unital quantum channel on $\mathcal M_n(\mathbb C)$ defined as in equation \eqref{eq:def-Phi}. Suppose that $d\geq(\log n)^{4+\epsilon}$ for some $\epsilon>0$. Then,
	\[ \E\left\|\Phi-\Pi\right\|_\infty \leq \frac{2}{\sqrt{d}}\left(1+\frac{C}{(\log n)^{\epsilon/6}}\right), \]
	where $C<\infty$ is an absolute constant. What is more, 
	\[ \P\left( \left\|\Phi-\Pi\right\|_\infty \leq \frac{2}{\sqrt{d}}\left(1+\frac{C'}{(\log n)^{\epsilon/6}}\right) \right) \geq 1-\frac{1}{n}, \]
	where $C'<\infty$ is an absolute constant.
\end{corollary}

The above result can be rephrased as follows: as soon as $d\gg (\log n)^4$, with probability at least $1-1/n$, the random unital quantum channel $\Phi$ on $\mathcal M_n(\mathbb C)$ is a $(d,2(1+\delta_n)/\sqrt{d})$ expander, with $\delta_n\to 0$ as $n\to\infty$. This means that it is an optimal expander, in the sense of equation \eqref{def:optimal-expander}.

Note also that, as discussed in the introduction, the conclusion of Corollary \ref{th:main-implication} can be equivalently written as 
\[ \P\left( s(\Phi) \leq \frac{2}{\sqrt{d}}\left(1+\frac{C'}{(\log n)^{\epsilon/6}}\right) \right) \geq 1-\frac{1}{n}, \]
and thus a fortiori 
\[ \P\left( \lambda(\Phi) \leq \frac{2}{\sqrt{d}}\left(1+\frac{C'}{(\log n)^{\epsilon/6}}\right) \right) \geq 1-\frac{1}{n}. \]
However, the latter statement is not expected to be tight, contrary to the former. Indeed, the conjecture is that $\lambda(\Phi)$ should typically be close to $1/\sqrt{d}$, not $2/\sqrt{d}$, as it was recently proven to be for Haar distributed unitaries \cite{Timhadjelt24}. But proving such result requires completely different tools from the ones used here, such as e.g.~Schwinger-Dyson equations.

\begin{proof}
	Corollary \ref{th:main-implication} is an immediate consequence of Theorem \ref{th:main}. Indeed,
	\[ \left\|\Phi-\Pi\right\|_\infty = \left\|M_\Phi-\Psi\right\|_\infty = \|X\|_\infty, \]
	for $X$ a random matrix on $\mathbb C^n\otimes\mathbb C^n$ defined as in equation \eqref{eq:def-X}. And Theorem \ref{th:main} gives us precisely an estimate on the average of the operator norm and the probability of deviating from it for such random matrix $X$.
\end{proof}

\section{Optimal \texorpdfstring{$k$}{k}-copy tensor product expanders from unitary \texorpdfstring{$2k$}{2k}-designs}
\label{sec:k-copy}

In this section, we move on to the general case, where $k$ is any fixed integer. Namely, we prove that a random unital quantum channel on $\mathcal M_{n^k}(\C)$ whose Kraus operators are sampled as $d$ independent $k$-copy tensor power unitaries from a $2k$-design is with high probability an optimal $k$-copy tensor product expander, in the sense of equation \eqref{def:optimal-expander}.

\subsection{Main technical results} \hfill\vspace{0.1cm}

Let $\mu$ be a $2k$-design on $U(n)$ and let $U\in U(n)$ be sampled from $\mu$. Take $U_1,\ldots,U_d\in U(n)$ independent copies of $U$ and set
\begin{equation} \label{eq:def-X-k} X:= \frac{1}{d}\sum_{s=1}^d \left(U_s^{\otimes k}\otimes\overbar{U}_s^{\otimes k}-\E(U_s^{\otimes k}\otimes\overbar{U}_s^{\otimes k})\right) = \frac{1}{d}\sum_{s=1}^d \left(U_s^{\otimes k}\otimes\overbar{U}_s^{\otimes k}-P^{(k)}\right), \end{equation}
where the last equality is because $\mu$ is in particular a $k$-design, and hence, for each $1\leq s\leq d$, $\E(U_s^{\otimes k}\otimes\overbar{U}_s^{\otimes k})=P^{(k)}$.

We begin with estimating the parameters $\sigma(X)$, $\upsilon(X)$ and $R(X)$ appearing in Theorem \ref{th:vanHandel} for the random matrix $X$ given by equation \eqref{eq:def-X-k}. This is the content of the following Lemmas \ref{lem:sigma-k}, \ref{lem:upsilon-k} and \ref{lem:Rp-k} respectively.

\begin{lemma} \label{lem:sigma-k}
	Let $X$ be a random matrix on $(\mathbb C^n)^{\otimes 2k}$ defined as in equation \eqref{eq:def-X-k}. Then,
	\[ \left\|\E(XX^*)\right\|_\infty = \left\|\E(X^*X)\right\|_\infty = \frac{1}{d}. \]
\end{lemma}

\begin{proof}
	Observe that
	\begin{align*} 
		\E(XX^*) & = \frac{1}{d^2}\sum_{s,t=1}^d \E\left(\left(U_s^{\otimes k}\otimes\overbar{U}_s^{\otimes k}-P^{(k)}\right) \left(U_t^{\otimes k}\otimes\overbar{U}_t^{\otimes k}-P^{(k)}\right)^*\right) \\
		& = \frac{1}{d^2}\sum_{s=1}^d \E\left(\left(U_s^{\otimes k}\otimes\overbar{U}_s^{\otimes k}-P^{(k)}\right) \left(U_s^{\otimes k}\otimes\overbar{U}_s^{\otimes k}-P^{(k)}\right)^*\right) \\
		& = \frac{1}{d}\E\left(\left(U^{\otimes k}\otimes\overbar{U}^{\otimes k}-P^{(k)}\right) \left(U^{\otimes k}\otimes\overbar{U}^{\otimes k}-P^{(k)}\right)^*\right), 
	\end{align*}
	where the second equality is because, for $1\leq s\neq t\leq d$, by independence of $U_s$ and $U_t$,
	\[ \E\left(\left(U_s^{\otimes k}\otimes\overbar{U}_s^{\otimes k}-P^{(k)}\right) \left(U_t^{\otimes k}\otimes\overbar{U}_t^{\otimes k}-P^{(k)}\right)^*\right) = \E\left(U_s^{\otimes k}\otimes\overbar{U}_s^{\otimes k}-P^{(k)}\right) \E\left(U_t^{\otimes k}\otimes\overbar{U}_t^{\otimes k}-P^{(k)}\right)^* = 0. \]
	Now, we have
	\begin{align*} 
		& \E\left(\left(U^{\otimes k}\otimes\overbar{U}^{\otimes k}-P^{(k)}\right) \left(U^{\otimes k}\otimes\overbar{U}^{\otimes k}-P^{(k)}\right)^*\right) \\
		& \quad = \E\left((U^{\otimes k}\otimes\overbar{U}^{\otimes k}) (U^{\otimes k}\otimes\overbar{U}^{\otimes k})^*\right)-\E\left(U^{\otimes k}\otimes\overbar{U}^{\otimes k}\right)P^{(k)} - P^{(k)}\E\left(U^{\otimes k}\otimes\overbar{U}^{\otimes k}\right)^* + P^{(k)} \\
		& \quad = I-P^{(k)}, 
	\end{align*}
	where the last equality is because $\E(U^{\otimes k}\otimes\overbar{U}^{\otimes k})=\E(U^{\otimes k}\otimes\overbar{U}^{\otimes k})^*=P^{(k)}$, while $(U^{\otimes k}\otimes\overbar{U}^{\otimes k}) (U^{\otimes k}\otimes\overbar{U}^{\otimes k})^*=I$. Hence,
	\[ \E(XX^*) = \frac{1}{d}\left(I-P^{(k)}\right), \]
	and we indeed have 
	\[ \left\|\E(XX^*)\right\|_\infty = \frac{1}{d}. \]
	The argument is exactly the same for $\|\E(X^*X)\|_\infty$.
\end{proof}

\begin{lemma} \label{lem:upsilon-k}
	Let $X$ be a random matrix on $(\mathbb C^n)^{\otimes 2k}$ defined as in equation \eqref{eq:def-X-k}. Then,
	\[ \left\|\Cov(X)\right\|_\infty \leq  \frac{1}{d}\times\frac{(C_0k^4)^{k}}{n}, \]
	where $C_0<\infty$ is an absolute constant.
\end{lemma}

\begin{proof}
	Observe that, denoting by $R$ the realignment operation between tensor factors $\{1,\dots,2k\}$ and $\{2k+1,\ldots,4k\}$,
	\begin{align*} 
		\Cov(X) & = \E\left(X\otimes\overbar{X}\right)^R \\
		& = \frac{1}{d^2} \sum_{s,t=1}^d \E\left(\left(U_s^{\otimes k}\otimes\overbar{U}_s^{\otimes k}-P^{(k)}\right)\otimes \left(\overbar{U}_t^{\otimes k}\otimes U_t^{\otimes k}-P^{(k)}\right)\right)^R \\
		& = \frac{1}{d^2} \sum_{s=1}^d \E\left(\left(U_s^{\otimes k}\otimes\overbar{U}_s^{\otimes k}-P^{(k)}\right)\otimes \left(\overbar{U}_s^{\otimes k}\otimes U_s^{\otimes k}-P^{(k)}\right)\right)^R \\
		& = \frac{1}{d} \E\left(\left(U^{\otimes k}\otimes\overbar{U}^{\otimes k}-P^{(k)}\right)\otimes \left(\overbar{U}^{\otimes k}\otimes U^{\otimes k}-P^{(k)}\right)\right)^R,
	\end{align*}
	where the third equality is because, for $1\leq s\neq t\leq d$, by independence of $U_s$ and $U_t$,
	\begin{align*}
		\E\left(\left(U_s^{\otimes k}\otimes\overbar{U}_s^{\otimes k}-P^{(k)}\right)\otimes \left(\overbar{U}_t^{\otimes k}\otimes U_t^{\otimes k}-P^{(k)}\right)\right)^R & = \left(\E\left(U_s^{\otimes k}\otimes\overbar{U}_s^{\otimes k}-P^{(k)}\right)\otimes \E\left(\overbar{U}_t^{\otimes k}\otimes U_t^{\otimes k}-P^{(k)}\right)\right)^R \\
		& = 0. 
	\end{align*}
	Now, we have
	\begin{align*} 
		& \E\left(\left(U^{\otimes k}\otimes\overbar{U}^{\otimes k}-P^{(k)}\right)\otimes \left(\overbar{U}^{\otimes k}\otimes U^{\otimes k}-P^{(k)}\right)\right) \\
		& \quad = \E\left(U^{\otimes k}\otimes\overbar{U}^{\otimes k}\otimes\overbar{U}^{\otimes k}\otimes U^{\otimes k}\right)-\E\left(U^{\otimes k}\otimes\overbar{U}^{\otimes k}\right)\otimes P^{(k)} - P^{(k)}\otimes\E\left(\overbar{U}^{\otimes k}\otimes U^{\otimes k}\right) + P^{(k)}\otimes P^{(k)} \\
		& \quad = \E\left(U^{\otimes k}\otimes\overbar{U}^{\otimes k}\otimes\overbar{U}^{\otimes k}\otimes U^{\otimes k}\right)-P^{(k)}\otimes P^{(k)}, 
	\end{align*}
	where the last equality is because $\E(U^{\otimes k}\otimes\overbar{U}^{\otimes k})=\E(\overbar{U}^{\otimes k}\otimes U^{\otimes k})=P^{(k)}$. Next, denoting by $S$ the swap operation between tensor factors $\{k+1,\ldots,2k\}$ and $\{3k+1,\ldots,4k\}$ (leaving tensor factors $\{1,\ldots,k\}$ and $\{2k+1,\ldots,3k\}$ unchanged), we have by equation \eqref{eq:P^k-Weingarten}
	\[ \E\left(U^{\otimes k}\otimes\overbar{U}^{\otimes k}\otimes\overbar{U}^{\otimes k}\otimes U^{\otimes k}\right) = \E\left(U^{\otimes 2k}\otimes\overbar{U}^{\otimes 2k}\right)^{S} = \sum_{\pi,\sigma\in\mathcal S_{2k}} \text{Wg}(n,\pi\sigma^{-1}) \ketbra{v_\pi}{v_\sigma}, \]
	where, given $\pi\in\mathcal S_{2k}$, the vector $v_\pi$ in $(\C^n)^{\otimes 2k}$ is defined as
	\[ \ket{v_\pi} = \sum_{1\leq i_1,\ldots,i_{2k}\leq n} \ket{i_1\cdots i_k i_{\pi(k+1)}\cdots i_{\pi(2k)}i_{\pi(1)}\cdots i_{\pi(k)}i_{k+1}\cdots i_{2k}} . \]
	And therefore,
	\[ \E\left(U^{\otimes k}\otimes\overbar{U}^{\otimes k}\otimes\overbar{U}^{\otimes k}\otimes U^{\otimes k}\right)^R = \sum_{\pi,\sigma\in\mathcal S_{2k}} \text{Wg}(n,\pi\sigma^{-1}) M_\pi\otimes M_\sigma, \]
	where, given $\pi\in\mathcal S_{2k}$, the matrix $M_\pi$ on $(\C^n)^{\otimes 2k}$ is defined as 
	\[ M_\pi = \sum_{1\leq i_1,\ldots,i_{2k}\leq n} \ketbra{i_1\cdots i_k i_{\pi(k+1)}\cdots i_{\pi(2k)}}{i_{\pi(1)}\cdots i_{\pi(k)} i_{k+1}\cdots i_{2k}}.  \]
	On the other hand, we have by equation \eqref{eq:P^k-Weingarten} again
	\[ \left(P^{(k)}\otimes P^{(k)}\right)^R = \sum_{\pi_1,\sigma_1,\pi_2,\sigma_2\in\mathcal S_{k}} \text{Wg}(n,\pi_1\sigma_1^{-1})\text{Wg}(n,\pi_2\sigma_2^{-1}) \ketbra{u_{\pi_1}}{u_{\pi_2}}\otimes\ketbra{u_{\sigma_1}}{u_{\sigma_2}}, \]
	where, given $\pi\in\mathcal S_{k}$, the vector $u_\pi$ in $(\C^n)^{\otimes k}$ is defined as
	\[ \ket{u_\pi} = \sum_{1\leq i_1,\ldots,i_{k}\leq n} \ket{i_1\cdots i_k i_{\pi(1)}\cdots i_{\pi(k)}} = n^k\ket{\psi_{1\pi(1)}\otimes\cdots\otimes\psi_{k\pi(k)}} . \]
	
	Let us now make a few observations. First, for any $\pi\in\mathcal S_{2k}$, 
	\begin{equation} \label{eq:norm-M_pi}
		\|M_\pi\|_\infty\leq\|M_\pi\|_2=n^k,
	\end{equation}
	where the equality in equation \eqref{eq:norm-M_pi} is because $M_\pi$ has exactly $n^{2k}$ entries equal to $1$ and all the others equal to $0$. Second,
	if $\pi\in\mathcal S_{2k}$ is of the form $\pi=\pi_1\pi_2$ with $\pi_1:\{k+1,\ldots,2k\}\to\{1,\ldots,k\}$ and $\pi_2:\{1,\ldots,k\}\to\{k+1,\ldots,2k\}$, then $M_\pi=\ketbra{u_{\pi_1}}{u_{\pi_2}}$, where we have identified $\pi_1,\pi_2$ as elements of $\mathcal S_k$. What is more, in the latter case, we clearly have $|\pi|=|\pi_1|+|\pi_2|$ and $\text{Mb}(\pi)=\text{Mb}(\pi_1)\text{Mb}(\pi_2)$. So by equation \eqref{eq:Weingarten-Moebius}, the dominating terms in $\text{Wg}(n,\pi)$ and $\text{Wg}(n,\pi_1)\text{Wg}(n,\pi_2)$ (which are equal to $\text{Mb}(\pi)/n^{2k+|\pi|}$ and $\text{Mb}(\pi_1)\text{Mb}(\pi_2)/n^{k+|\pi_1|+k+|\pi_2|}$ respectively) cancel out, and we have $\text{Wg}(n,\pi)-\text{Wg}(n,\pi_1)\text{Wg}(n,\pi_2)=O(1/n^{2k+|\pi|+2})$. Hence as a consequence, we can write
	\begin{equation} \label{eq:expectation-k-R} \left(\E\left(U^{\otimes k}\otimes\overbar{U}^{\otimes k}\otimes\overbar{U}^{\otimes k}\otimes U^{\otimes k}\right)-P^{(k)}\otimes P^{(k)}\right)^R = \sum_{\substack{\pi,\sigma\in\mathcal S_{2k} \\ \pi\neq\pi_1\pi_2\,\text{or}\,\sigma\neq\sigma_1\sigma_2}} \text{Wg}(n,\pi\sigma^{-1}) M_\pi\otimes M_\sigma + R , \end{equation}
	where we have defined
	\[ R = \sum_{\pi_1,\sigma_1,\pi_2,\sigma_2\in\mathcal S_{k}} \left( \text{Wg}(n,\pi_1\sigma_1^{-1}\pi_2\sigma_2^{-1}) - \text{Wg}(n,\pi_1\sigma_1^{-1})\text{Wg}(n,\pi_2\sigma_2^{-1}) \right)\ketbra{u_{\pi_1}}{u_{\pi_2}}\otimes\ketbra{u_{\sigma_1}}{u_{\sigma_2}}, \]
	so that $\|R\|_\infty\leq C(k!)^4/n^2$.
	
	Next, note that there is equality in the inequality in equation \eqref{eq:norm-M_pi} if and only if $M_\pi$ has rank $1$, which happens if and only if $\pi$ is of the form $\pi=\pi_1\pi_2$, and thus $M_\pi=\ketbra{u_{\pi_1}}{u_{\pi_2}}$. We actually claim that, if $\pi\neq\pi_1\pi_2$, then $\|M_\pi\|_\infty\leq n^{k-1}$. Indeed, suppose that there exists $l\leq k-1$ such that $\pi$ has exactly $l$ elements of $\{1,\ldots,k\}$ that have elements of $\{k+1,\ldots,2k\}$ as images, and hence as well exactly $l$ elements of $\{k+1,\ldots,2k\}$ that have elements of $\{1,\ldots,k\}$ as images. Then, there exist $x_1,\ldots,x_l$, resp.~$x_1',\ldots,x_l'$, distinct elements of $\{1,\ldots,k\}$, resp.~$\{k+1,\ldots,2k\}$, and $\alpha$, resp.~$\beta$, permutation of $\{1,\ldots,k\}$, resp.~$\{k+1,\ldots,2k\}$, such that
	\[ M_\pi = \left(U^{(\alpha)}_{\{1,\ldots,k\}}\otimes U^{(\beta)}_{\{k+1,\ldots,2k\}}\right)\left(n^l \Psi_{x_1x_1'}\otimes\cdots\otimes\Psi_{x_lx_l'}\otimes I_{\{1,\ldots,2k\}\setminus\{x_1,x_1',\ldots,x_l,x_l'\}}\right), \]
	where $U^{(\alpha)},U^{(\beta)}$ denote the permutation unitaries associated to $\alpha,\beta$. We thus have by unitary invariance 
	\[ \left\|M_\pi\right\|_\infty = n^l \left\|\Psi^{\otimes l}\otimes I^{\otimes 2(k-l)}\right\|_\infty = n^l \leq n^{k-1}. \]
	This implies that, for any $\pi,\sigma\in\mathcal S_{2k}$ such that $\pi\neq\pi_1\pi_2$ or $\sigma\neq\sigma_1\sigma_2$, $\|M_\pi\otimes M_\sigma\|_\infty\leq n^{2k-1}$.
	Using this observation, and recalling that, for any $\pi,\sigma\in\mathcal S_{2k}$, $\text{Wg}(n,\pi\sigma^{-1})=O(1/n^{2k})$, we get from equation \eqref{eq:expectation-k-R} that
	\[ \left\|\left(\E\left(U^{\otimes k}\otimes\overbar{U}^{\otimes k}\otimes\overbar{U}^{\otimes k}\otimes U^{\otimes k}\right)-P^{(k)}\otimes P^{(k)}\right)^R\right\|_\infty \leq \frac{C'((2k)!)^2}{n} + \frac{C(k!)^4}{n^2} \leq \frac{(C''k^4)^{k}}{n} . \]

	Putting everything together, we indeed obtain
	\[ \left\|\Cov(X)\right\|_\infty \leq \frac{1}{d}\times\frac{(C''k^4)^{k}}{n}. \]
\end{proof}

\begin{lemma} \label{lem:Rp-k}
	Let $X=\sum_{s=1}^d Z_s$ be a random matrix on $(\mathbb C^n)^{\otimes 2k}$ defined as in equation \eqref{eq:def-X-k}, where $Z_s=(U_s^{\otimes k}\otimes\overbar{U}_s^{\otimes k}-P^{(k)})/d$, $1\leq s\leq d$. Then,
	\[ \max_{1\leq s\leq d} \left\|Z_s\right\|_\infty = \frac{1}{d}. \]
\end{lemma}

\begin{proof}
	Observe that
	\[ \max_{1\leq s\leq d} \left\|Z_s\right\|_\infty = \frac{1}{d}\,\max_{1\leq s\leq d} \left\|U_s^{\otimes k}\otimes\overbar{U}_s^{\otimes k}-P^{(k)}\right\|_\infty = \frac{1}{d}\left\|U^{\otimes k}\otimes\overbar{U}^{\otimes k}-P^{(k)}\right\|_\infty. \]
	Now, we have
	\begin{align*}
		\left\|U^{\otimes k}\otimes\overbar{U}^{\otimes k}-P^{(k)}\right\|_\infty^{1/2} & = \left\|\left(U^{\otimes k}\otimes\overbar{U}^{\otimes k}-P^{(k)}\right)\left(U^{\otimes k}\otimes\overbar{U}^{\otimes k}-P^{(k)}\right)^*\right\|_\infty \\
		& = \left\|(UU^*)^{\otimes k}\otimes(\overbar{U}\overbar{U}^*)^{\otimes k}-(U^{\otimes k}\otimes\overbar{U}^{\otimes k})P^{(k)}-P^{(k)}(U^{\otimes k}\otimes\overbar{U}^{\otimes k})^*+P^{(k)}\right\|_\infty \\
		& = \left\|I-P^{(k)}\right\|_\infty \\
		& = 1,
	\end{align*}
	where the third equality is because $(U^{\otimes k}\otimes\overbar{U}^{\otimes k})P^{(k)}=P^{(k)}(U^{\otimes k}\otimes\overbar{U}^{\otimes k})^*=P^{(k)}$ and $UU^*=\overbar{U}\overbar{U}^*=I$. So we indeed have
	\[ \max_{1\leq s\leq d} \left\|Z_s\right\|_\infty = \frac{1}{d}. \]
\end{proof}

Just as in the case $k=1$, we point out that, for the results of Lemmas \ref{lem:sigma-k} and \ref{lem:Rp-k} to hold, it is actually enough that the operators $U_1,\ldots,U_d\in U(n)$ are independent unitaries, without any extra assumption on their distribution. It is only for Lemma \ref{lem:upsilon-k} that we need them to be sampled from a $2k$-design.

With Lemmas \ref{lem:sigma-k}, \ref{lem:upsilon-k} and \ref{lem:Rp-k} at hand, we are now ready to state the main result of this section.

\begin{theorem} \label{th:main-k}
	Let $X$ be a random matrix on $\mathbb C^n\otimes\mathbb C^n$ defined as in equation \eqref{eq:def-X-k}. Suppose that $d\geq(\log n)^{4+\epsilon}$ for some $\epsilon>0$. Then,
	\[ \E\|X\|_\infty \leq \frac{2}{\sqrt{d}}\left(1+\frac{C_k}{(\log n)^{\epsilon/6}}\right), \]
	where $C_k<\infty$ is a constant that depends only on $k$. What is more, 
	\[ \P\left( \|X\|_\infty \leq \frac{2}{\sqrt{d}}\left(1+\frac{C_k'}{(\log n)^{\epsilon/6}}\right) \right) \geq 1-\frac{1}{n^k}, \]
	where $C_k'<\infty$ is a constant that depends only on $k$.
\end{theorem}

\begin{proof}
	By Lemmas \ref{lem:sigma-k}, \ref{lem:upsilon-k} and \ref{lem:Rp-k}, respectively, we can estimate the parameters $\sigma(X)$, $\upsilon(X)$ and $R(X)$ defined in Theorem \ref{th:vanHandel}. We have
	\[ \sigma(X)=\frac{1}{\sqrt{d}}, \quad \upsilon(X)\leq\frac{1}{\sqrt{d}}\times\frac{(C_0k^4)^{k/2}}{\sqrt{n}}, \quad R(X)=\frac{1}{d}. \]
	Applying the first statement in Theorem \ref{th:vanHandel}, we thus get, setting $C_0'=C_0^{1/4}$,
	\[ \E\|X\|_\infty \leq \frac{2}{\sqrt{d}} + \frac{C}{\sqrt{d}}\left( \frac{(\log n)^{3/4}(C_0'k)^k}{n^{1/4}} + \frac{(\log n)^{2/3}}{d^{1/6}} + \frac{\log n}{d^{1/2}} \right). \]
	Hence, for $d\geq(\log n)^{4+\epsilon}$,
	\[ \E\|X\|_\infty \leq \frac{2}{\sqrt{d}} + \frac{C}{\sqrt{d}}\times\frac{3(C_0'k)^k}{(\log n)^{\epsilon/6}}, \]
	which is exactly the first statement in Theorem \ref{th:main} (up to relabeling $3(C_0'k)^kC/2$ into $C_k$).
	
	To prove the second statement in Theorem \ref{th:main}, we just have to apply the second statement in Theorem \ref{th:vanHandel}, to obtain, setting again $C_0'=C_0^{1/4}$,
	\[ \forall\ t>0,\ \P\left( \|X\|_\infty \leq \frac{2}{\sqrt{d}} + \frac{C}{\sqrt{d}}\left( \frac{(\log n)^{3/4}(C_0'k)^k}{n^{1/4}} + \frac{t^{2/3}}{d^{1/6}} + \frac{t}{d^{1/2}} + \frac{(C_0'k)^{2k}t^{1/2}}{n^{1/2}} \right) \right) \geq 1-4n^{2k}e^{-t}. \]
	Applying the above to $t=5k\log n$, we have that, for $d\geq(\log n)^{4+\epsilon}$,
	\[ \P\left( \|X\|_\infty \leq \frac{2}{\sqrt{d}} + \frac{C}{\sqrt{d}}\times\frac{4\times(5k)^{1/2}(C_0'k)^{2k}}{(\log n)^{\epsilon/6}} \right) \geq 1-\frac{1}{n^k}, \]
	which implies exactly the announced result (setting $C_k'=2\times(5k)^{1/2}(C_0'k)^{2k}C$).
\end{proof}

\subsection{Implication concerning random `mixture of tensor power unitaries' quantum channels} \hfill\vspace{0.1cm}

Let $\mu$ be a $2k$-design on $U(n)$ and let $U_1,\ldots,U_d\in U(n)$ be independently sampled from $\mu$. Define the random CP map
\begin{equation} \label{eq:def-Phi-k}
	\Phi:Y\in\mathcal M_{n^k}(\mathbb C)\mapsto \frac{1}{d}\sum_{s=1}^d U_s^{\otimes k}YU_s^{*\otimes k}\in\mathcal M_{n^k}(\mathbb C).
\end{equation}
$\Phi$ is by construction TP and unital (because $U_s^{*\otimes k}U_s^{\otimes k}=U_s^{\otimes k}U_s^{*\otimes k}=I$ for each $1\leq s\leq d$). Additionally, since $\mu$ is in particular a $k$-design, $\E(\Phi)=\Pi^{(k)}$.

\begin{corollary} \label{th:main-implication-k}
	Let $\Phi$ be a random unital quantum channel on $\mathcal M_{n^k}(\mathbb C)$ defined as in equation \eqref{eq:def-Phi-k}. Suppose that $d\geq(\log n)^{4+\epsilon}$ for some $\epsilon>0$. Then,
	\[ \E\left\|\Phi-\Pi^{(k)}\right\|\infty \leq \frac{2}{\sqrt{d}}\left(1+\frac{C_k}{(\log n)^{\epsilon/6}}\right), \]
	where $C_k<\infty$ is a constant that depends only on $k$. What is more, 
	\[ \P\left( \left\|\Phi-\Pi^{(k)}\right\|_\infty \leq \frac{2}{\sqrt{d}}\left(1+\frac{C_k'}{(\log n)^{\epsilon/6}}\right) \right) \geq 1-\frac{1}{n^k}, \]
	where $C_k'<\infty$ is a constant that depends only on $k$.
\end{corollary}

The above result can be rephrased as follows: for $k$ fixed and as soon as $d\gg (\log n)^4$, with probability at least $1-1/n^k$, the random unital quantum channel $\Phi$ on $\mathcal M_{n^k}(\mathbb C)$ is a $(d,2(1+\delta_n)/\sqrt{d})$ $k$-copy tensor product expander, with $\delta_n\to 0$ as $n\to\infty$. This means that it is, in particular, an optimal expander, in the sense of equation \eqref{def:optimal-expander} (whose Kraus operators additionally have the property of being of $k$-copy tensor power form).

Note also that, as discussed in the introduction, the conclusion of Corollary \ref{th:main-implication-k} can be equivalently written as 
\[ \P\left( s(\Phi) \leq \frac{2}{\sqrt{d}}\left(1+\frac{C_k'}{(\log n)^{\epsilon/6}}\right) \right) \geq 1-\frac{1}{n^k}, \]
and thus a fortiori 
\[ \P\left( \lambda(\Phi) \leq \frac{2}{\sqrt{d}}\left(1+\frac{C_k'}{(\log n)^{\epsilon/6}}\right) \right) \geq 1-\frac{1}{n^k}. \]

\begin{proof}
	Corollary \ref{th:main-implication-k} is an immediate consequence of Theorem \ref{th:main-k}. Indeed,
	\[ \left\|\Phi-\Pi^{(k)}\right\|_\infty = \left\|M_\Phi-P^{(k)}\right\|_\infty = \|X\|_\infty, \]
	for $X$ a random matrix on $\mathbb C^{n^k}\otimes\mathbb C^{n^k}$ defined as in equation \eqref{eq:def-X-k}. And Theorem \ref{th:main-k} gives us precisely an estimate on the average of the operator norm and the probability of deviating from it for such random matrix $X$.
\end{proof}

\section{Summary and perspectives}
\label{sec:summary}

One of the main open questions around quantum expanders is to come up with explicit constructions of optimal ones (or in fact even just close to optimal ones). Indeed, all known examples of $(d,C/\sqrt{d})$ expanders are random constructions. This work can be seen as a step towards derandomization, since it shows that sampling unitary Kraus operators from a much simpler distribution than the Haar measure, namely a $2$-design, already provides on optimal expander. For instance, we can now assert that a random unital quantum channel on $(\C^2)^{\otimes N}$ whose Kraus operators are $d\geq\mathrm{poly}(N)$ unitaries picked independently and uniformly from the finite $N$-qubit Clifford group is with high probability an optimal expander. It could thus be that exhaustive search for an explicit example of optimal expander is getting within reach. One question that remains unanswered is whether the same would be true for a random unital quantum channel on $\C^n$ whose Kraus operators are $d\geq\mathrm{poly}(\log n)$ unitaries picked independently and uniformly from the finite $d$-dimensional generalized Pauli group. Indeed, the latter is only a $1$-design, not a $2$-design, so our result does not apply. It is however not clear whether the need for imposing a $2$-design condition is just an artifact of the proof strategy (in order to control the covariance matrix) or whether there is a true obstruction for the result to hold with only a $1$-design condition. In fact, the same question can be asked in the general case: is sampling from a $2k$-design really necessary or could sampling from a $k$-design be enough in order to obtain with high probability an optimal $k$-copy tensor product expander?

Another natural question would be whether our results still hold when sampling the Kraus operators from $\mu$ an approximate rather than exact $2k$-design on $U(n)$. When we look into the proofs, we see that what we actually need are upper bounds of the form $\|P_\mu^{(k)}-P^{(k)}\|_\infty\leq C$ and $\|(P_\mu^{(2k)}-P^{(2k)})^R\|_\infty,\|(P^{(k)}\otimes[P_\mu^{(k)}-P^{(k)}])^R\|_\infty\leq C'/m_n$ with $m_n\gg(\log n)^3$. Indeed, if these conditions hold, we can conclude that the corresponding random unital quantum channel is typically a $(d,C''/\sqrt{d})$ $k$-copy tensor product expander, i.e.~an optimal $k$-copy tensor product expander up to a constant multiplicative factor. Now, it was recently proven in \cite{HHS24} (following a series of works over the past decade \cite{Brandao16,Harrow23,Haferkamp22}) that, for various distributions $\mu_{N,\ell}$ of random depth-$\ell$ $N$-qubit circuits, we have $\|P_{\mu_{N,\ell}}^{(k)}-P^{(k)}\|_\infty\leq\varepsilon$ as soon as $\ell\geq \log(N/\varepsilon)\mathrm{poly}(k)$. Replacing $\varepsilon$ by $\varepsilon/4^{Nk}$ in the above approximation result (and applying it to $2k$ instead of $k$) guarantees that, for $\ell\geq \mathrm{poly}(N,k)\log(1/\varepsilon)$, we also have $\|(P_{\mu_{N,\ell}}^{(2k)}-P^{(2k)})^R\|_\infty,\|(P^{(k)}\otimes[P_{\mu_{N,\ell}}^{(k)}-P^{(k)}])^R\|_\infty\leq\varepsilon$. This is because the realignment on $\C^{4^{Nk}}\otimes\C^{4^{Nk}}$ can increase the $\infty$-norm by at most $4^{Nk}$. Hence, our peculiar notion of approximation is indeed fulfilled by distributions that can be efficiently generated, namely random circuits of depth $\mathrm{poly}(\log n,k)$. This is one more argument proving that our random constructions of optimal tensor product expanders are easy to sample, and hence interesting in practice.

In this work, we have chosen to consider only the regime where the number of copies $k$ is fixed while the local dimension $n$ grows. We have therefore not made any effort to try and optimize the dependence on $k$ of the constants $C_k,C_k'$ appearing in Theorem \ref{th:main-k} and Corollary \ref{th:main-implication-k}. When looking into the proofs, we see that they scale respectively as $k^k,k^{2k+1/2}$. This means that our results as they are remain interesting as long as $k$ grows slower than $\log\log n$. This could potentially be improved by dealing less roughly with the combinatorics of permutations in the proof of Lemma \ref{lem:upsilon-k}.

As a final, more general, observation, note that this work's overall philosophy has already been followed before, in slightly different contexts. It basically consists in asking: knowing that a random construction based on Haar distributed unitaries satisfies with high probability a given property of interest, is the same true when sampling those unitaries according to a more easily implementable approximation of the Haar measure, such as a $k$-design measure? Such partial derandomization question has been explored for problems very closely related to the one we consider in this work, namely: approximating $\Pi^{(k)}$ in other norms than the $2{\to}2$ norm that shows up here, such as the $1{\to}1$ or $1{\to}\infty$ norms (together with their completely bounded versions) that are for instance relevant in quantum cryptography \cite{Aubrun09,Lancien20}. It has also been looked at in relation to non-additivity problems in quantum information, with the main technical tool being a version of Dvoretzky's theorem valid for $k$-design rather than Haar distributed random subspaces \cite{Nema22}.

\subsection*{Funding} At the time this research occurred, I was supported by the project ESQuisses of the Agence Nationale de la Recherche (grant number ANR-20-CE47-0014-01).

\subsection*{Acknowledgments} I would like to thank Guillaume Aubrun for the first blackboard discussions that launched this investigation. These took place during a research term at ICMAT (Madrid, Spain) so I am also thankful to the institution and the organizers for the hospitality. Finally, I am grateful to Motohisa Fukuda for prompting me to wrap up the draft I had left hanging, and for the nice discussions around it afterwards. At the time this research occurred, I was supported by the ANR project ESQuisses (grant number ANR-20-CE47-0014-01).

\end{document}